\documentclass[a4paper,reqno]{article}
\usepackage{comment}
\usepackage{amssymb}
\usepackage{latexsym}
\usepackage{amsmath}
\usepackage{graphicx}
\usepackage{relsize}
\usepackage{amsthm}
\usepackage{empheq}
\usepackage{bm}
\usepackage{booktabs}
\usepackage[dvipsnames]{xcolor}
\usepackage{pagecolor}
\usepackage{subcaption}
\usepackage{tikz,pgfplots,lipsum,lmodern}
\usepackage{enumitem}
\usepackage[most]{tcolorbox}

\newcommand{\dstirling}[2]{\genfrac{[}{]}{0pt}{0}{#1}{#2}}
\usepackage[margin=2.5cm]{geometry}

 \usepackage[section]{placeins}

\numberwithin{equation}{section}

\theoremstyle{definition}
\newtheorem{theorem}{Theorem}[section]
\newtheorem{corollary}[theorem]{Corollary}
\newtheorem{proposition}[theorem]{Proposition}

\newtheorem{lemma}[theorem]{Lemma}

\usepackage{calrsfs}

\usepackage{calrsfs}
\DeclareMathAlphabet{\pazocal}{OMS}{zplm}{m}{n}

\newcommand{\numberset}{\mathbb}

\newcommand{\R}{\numberset{R}}

\newcommand{\mV}{\mathcal{V}}

\newcommand{\mR}{\pazocal{R}}

\newcommand{\mC}{\mathcal{C}}

\newcommand{\mF}{\mathcal{F}}
\newcommand{\mW}{\mathcal{W}}

\newcommand{\mB}{\mathcal{B}}
\newcommand{\mE}{\mathcal{E}}

\renewcommand{\longrightarrow}{\to}

\newcommand*{\myproofname}{Proof of the claim}

\makeatletter
\renewcommand*\env@matrix[1][*\c@MaxMatrixCols c]{%
  \hskip -\arraycolsep
  \let\@ifnextchar\new@ifnextchar
  \array{#1}}
\makeatother

\usepackage{hyperref}

\title{
\textbf{On the Density of Codes over Finite Chain Rings }}

\usepackage{authblk}

\author[1]{Anna-Lena Horlemann}
\affil[1]{University of St.Gallen, Switzerland}
\author[2]{Violetta Weger}
\affil[2]{Technical University of Munich, Germany}
\author[1]{Nadja Willenborg}

\date{}

\usepackage{setspace}

\begin{document}

\maketitle
	
\thispagestyle{empty}
	
\begin{abstract}
We determine the asymptotic proportion of free modules over finite chain rings with good distance properties and treat the asymptotics in the code length $n$ and the residue field size $q$ separately. We then specialize and apply our technique to rank metric codes and to Hamming metric codes.
\end{abstract}

\section{Introduction}
The study of the asymptotic performance of error-correcting codes is a classical topic in coding theory, going back to the work by Shannon, where random codes are used to get arbitrarily close to the capacity of a given discrete memoryless channel \cite{shannon48}. It is well known that Hamming-metric codes over finite fields attaining the Singleton bound, namely Maximum Distance Separable (MDS) codes, are dense if the field size tends to infinity and, assuming the validity of the MDS conjecture, that they are sparse for the length tending to infinity. For rank metric codes, we call codes that attain the Singleton bound Maximum Rank Distance (MRD) codes. It makes a great difference whether we have $\mathbb{F}_q$-linear codes, or $\mathbb{F}_{q^m}$- linear codes. For the former, we have that MRD codes are neither dense nor sparse as $q$ tends to infinity, while for the latter they are dense \cite{heide, alb, b11}. Recently, also the Lee metric was investigated in \cite{b1,lee}, and Maximum Lee Distance (MLD) codes are sparse, whether we let $n$ or $q$ grow to infinity. In this paper we consider such questions in the context of linear codes over finite chain rings, i.e., rings whose ideals are generated by one element and are linearly ordered by inclusion. Research on codes over rings has been started a while ago, see e.g. \cite{b2,honold2000linear}.
Recently, there has been increasing interest in rank-metric codes over finite rings, see for example the literature on nested-lattice based network coding \cite{b13, b14}, which shows that by using network coding over finite rings one may develop more efficient physical layer network coding schemes.

\section{Preliminaries}
  Throughout this paper we denote by $Q$ the set of prime powers, by $\mR$  a \emph{finite chain ring} and the generator of its maximal ideal with nilpotency index $s$ is denoted by $\gamma$. Assume that for the residue field we have $\mR/ \langle \gamma  \rangle  \cong \mathbb{F}_{q}$, i.e., $|\mR|=q^s$, the projection $\mR \rightarrow \mR/ \langle \gamma  \rangle$ is denoted by $\Psi$ and can be naturally extended to polynomials over $\mR$. 
    
    Let $h \in \mR[X]$ be a polynomial of degree $m$ such that the leading coefficient of $h$ is a unit and $\Psi(h)$ is irreducible in $\mathbb{F}_{q}$, then
$S= \mR[X]/\langle h \rangle$
is a \emph{free local degree-$m$ Galois extension} of $\mR$ with maximal ideal $\gamma S$ and $S/ \langle \gamma  \rangle \cong \mathbb{F}_{q^{m}}$. Since $\mR$ is a finite chain ring, $S$ is also a finite chain ring. In particular, it is a free $\mR$-module of rank $m$ and hence letting $N=nm$, the $S$-module $S^{n}$ is isomorphic to $\mR^{N}$. 
For every divisor $\ell$ of $m$, we have  a unique subring $\Bar{S} \subseteq S$ which is a free local Galois extension of degree $\ell$ of $\mR$. For more details, we refer to \cite{b6}. 

    We consider \emph{codes} as subsets of the metric space $(S^{n},D)$, where $D$ is a \emph{translation-invariant distance function} $D: S^{n} \times S^{n} \rightarrow \R_{\ge 0}$ which is non-increasing under scalar multiplication. (This property is required for the technique that we use to count the number of codes, and unfortunately excludes some metrics, such as the Lee metric).
    The \emph{minimum distance} of a code $\mC \subseteq S^{n}$ with respect to the distance $D$ is defined as
\begin{align*}
    D(\mC) := \min\{D(x,y) : x,y \in \mC, \, x \ne y\}.
\end{align*}
   Similarly to finite fields, we can study codes that are linear w.r.t. the subring $\Bar{S} \subseteq S$. These codes are then finitely generated $\Bar{S}$-submodules of $S^n$.
   The smallest number of elements in $\mC$ which generate $\mC$ as an $\Bar{S}$-module is denoted by $\mu_{\Bar{S}}(\mC)$ and called the \emph{rank} of $\mC$. If $\mC = \{0\}$, then we set $\mu_{\Bar{S}}(\mC)=0$. If $\mC$ is a finitely generated free $\Bar{S}$-module and $\{e_{1}, \dots , e_{k}\}$ is a free basis of $\mC$, then $\mu_{\Bar{S}}(\mC)=k$ and any generating set of $\mC$ consisting of $k$ elements is a free basis of $\mC$ over $\Bar{S}$. 
    
  From now on, we let $S$ be a free local Galois extension of $\mR$ with fixed degree $m$ such that $S^{n} \cong \mR^{N}$ and we denote by $\Bar{S}$ its unique subring of extension degree $\ell$. Moreover a code $\mC$ is always considered to be a free $\Bar{S}$-module of $S^{n}$.
    We restrict ourselves to this case, due to the density results from \cite{b1} for modules over finite chain rings, which state that the probability of picking $k$ generators $e_{1}, \dots, e_{k} \in S^{n}$, all having at least one unit entry, approaches $1$. This property is independent of the underlying linearity of the code, but only determines the freeness of the code and hence we can apply it to our context. 
    From \cite{honold2000linear,b1} it is known that the number of free $\Bar{S}$-modules of $S^n$ having rank $k$ is 
\begin{equation} \label{eq:N}
    N_{N/\ell,q^{\ell}}^{k} := q^{(N-k\ell)k(s-1)}\dstirling{N/\ell}{k}_{q^{\ell}},
\end{equation}
where $q,\ell,s,N = mn$ are as above. Note that \eqref{eq:N} can define a sequence with respect to any of its parameters. We will often use this, without explicitly defining the sequence.
    We denote by
    \begin{equation*} 
    \delta_q^D(N,k,\ell, d) := \frac{|\{\mC \subseteq \mR^{N}  :  \mu_{\Bar{S}}(\mC) =k, \, D(\mC) \ge d\}|}{ N_{N/\ell,q^{\ell}}^{k}}
\end{equation*}
the \emph{density of $\Bar{S}$-linear codes} in $\mR^N$ having rank $k$ over $\Bar{S}$ and minimum distance at least $d$. In the following, $d$ will always be a finite positive integer, not exceeding the maximal possible distance with respect to $D$.

To give an upper and lower bound on the density $\delta_q^D(N,k,\ell, d)$ we use counting arguments on the number of isolated vertices in \emph{bipartite graphs} that are regular with respect to certain maps defined on their left-vertices. This approach was first proposed in \cite{alb} and applied to different metrics and various types of linearity in \cite{b5}. 

We recall the concept of an \emph{association} on a finite non-empty set $\mV$ of magnitude $r \geq 0$ (see \cite{alb}) which is a function 
$\alpha: \mV \times \mV \to \{0,...,r\}$ satisfying the following:
\begin{itemize}
\item[(i)] $\alpha(V,V)=r$ for all $V\in \mV$;
\item[(ii)] $\alpha(V,V')=\alpha(V',V)$ for all $V,V' \in \mV$.
\end{itemize}

Given a finite bipartite graph $\mB=(\mV,\mW,\mE)$  and an association $\alpha$ on~$\mV$ of magnitude $r$, we call $\mB$ \emph{$\alpha$-regular} if for all  $(V,V') \in \mV \times \mV$ the number of vertices $W \in \mW$ with $(V,W) \in \mE$ and 
$(V',W) \in \mE$ only depends on $\alpha(V,V')$. We then denote this number by~$\mW_\ell(\alpha)$, where $\ell=\alpha(V,V') \in \{0, \dots , r \}$.
 \begin{proposition}\cite[Lemma 3.2 \& Lemma 3.5]{alb} \label{lem:upperbound}
 \\
Let $\mB=(\mV,\mW,\mE)$ be a bipartite $\alpha$-regular graph, where $\alpha$ is an association on $\mV$ of magnitude $r$. 
Let $\mF \subseteq \mW$ be the collection of non-isolated vertices of $\mW$. If $\mW_{r}(\alpha) > 0$, then
\begin{enumerate}
    \item [(i)]$|\mF| \le |\mV| \, \mW_{r}(\alpha).$
\item[(ii)]$|\mF| \ge  \frac{\mW_r(\alpha)^2 \, |\mV|^2}{\sum_{\ell=0}^r  \mW_\ell(\alpha) \, |\alpha^{-1}(\ell)|}.$
\end{enumerate}
\end{proposition}
We often consider the set $S_{q(\gamma)}^{n}:= S^{n} \setminus (\gamma S)^{n}$, where the lower index indicates the dependence given through $S/\langle \gamma \rangle \cong \mathbb{F}_{q^{m}}$, equivalently $\mR/\langle \gamma \rangle \cong \mathbb{F}_{q}$.
We denote 
the \emph{open ball} in $S^{n}$ of radius $r$ centered at $0$ by 
$$\textbf{B}_{q}^{D}(S^{n},r):= \{x \in S^{n} :  D(x,0) < r \}$$
and the open ball of radius $r$ centered at $0$, containing only elements from $S_{q(\gamma)}^{n}$ by 
$$\textbf{B}_{q}^{D}(S_{q(\gamma)}^{n},r):= \textbf{B}_{q}^{D}(S^{n},r)\setminus  \textbf{B}_{q}^{D}((S/\langle \gamma^{s-1}\rangle)^{n},r),$$
where $\textbf{B}_{q}^{D}((S/\langle \gamma^{s-1}\rangle)^{n},r)$ describes the open ball of radius $r$ over the chain ring $S/\langle \gamma^{s-1}\rangle$, centered at $0$. Note that $0 \notin\textbf{B}_{q}^{D}(S_{q(\gamma)}^{n},r)$. Since we consider translation-invariant distances the size of any ball in $S^{n}$, or in any of its subrings, does not depend on its center.
Hence, we do not specify the centers and denote the corresponding \emph{volumes} of the balls by $\textbf{v}_{q}^{D}(S^{n},r)
, \textbf{v}_q^{D}(S/\langle \gamma^{s-1}\rangle^n,r)
$, respectively $\textbf{v}_{q}^{D}(S_{q(\gamma)}^{n},r)$. 

\section{Bounding the Density}
First we use Proposition \ref{lem:upperbound} to give an upper and a lower bound on the number of codes in the metric space $(S^{n}, D)$ that have minimum distance bounded from above by some positive integer $d$. In contrast to previous literature \cite{alb}, \cite{b5} we extend the bipartite graph model by considering free modules instead of vector spaces as vertices of the graph. A free module of rank $1$ is generated by an element from the set $S_{q(\gamma)}^{n}$. These modules will represent the left-sided vertices in our bipartite graph, which will therefore fulfill some important regularity properties. 

\begin{theorem} 
Let $1 \le k \le N/\ell$, $1 \le d < \infty$ and let
\begin{align*}
    \mF := \{\mC \subseteq \mR^{N}  :   D(\mC) \le d-1, \, \mu_{\Bar{S}}(\mC) = k\}.
\end{align*}
We have
\begin{align*}
    |\mF| &\le  \frac{\textbf{v}_q^{D}(S_{q(\gamma)}^{n},d)N_{N/\ell-1,q^{\ell}}^{k-1}}{q^{\ell(s-1)}(q^{\ell}-1)} , \\[0.2cm]
    |\mF| &\ge \tfrac{ \left(\frac{\textbf{v}_q^{D}(S_{q(\gamma)}^{n},d)}{q^{\ell(s-1)}(q^{\ell}-1)} \right)(N_{N/\ell-1,q^{\ell}}^{k-1})^{2}}{N_{N/\ell-1,q^{\ell}}^{k-1} + \left( \frac{\textbf{v}_q^{D}(S_{q(\gamma)}^{n},d)}{q^{\ell(s-1)}(q^{\ell}-1)}-1 \right)N_{N/\ell-2,q^{\ell}}^{k-2}}.
\end{align*}
\end{theorem}
\begin{proof}
Consider the bipartite graph $\mB=(\mV,\mW,\mE),$ where $\mV$ is the set of modules generated by one element from $\textbf{B}_q^{D}(S_{q(\gamma)}^{n},d)$, $\mW$ is the collection of free codes in $\mR^N$ with $\mu_{\bar{S}}(\mC)=k$ and $(V,\mC) \in \mE$ if and only if $V 
\subseteq \mC$. Every vertex contains exactly $q^{\ell(s-1)}(q^{\ell}-1)$ elements from $\Bar{S}_{\gamma}^{n}$. Hence we have
\begin{align*}
    |\mV| = \frac{\textbf{v}_q^{D}(S_{q(\gamma)}^{n},d)}{q^{\ell(s-1)}(q^{\ell}-1)}, \quad  |\mW| = N_{N/\ell,q^{\ell}}^{k}.
\end{align*}
Due to the assumption, that $D$ is non-increasing under scalar multiplication, every $x \in \textbf{B}_q^{D}(\mR^{N},d)$ belongs to a vertex $V \in \mV$ and the set of non-isolated vertices in $\mW$ is exactly $\mF$.

We define the association
$$\alpha  :  \mV \times \mV \longrightarrow \{0,1\}, \quad (V,V') \mapsto 2-\mu_{\bar{S}}(\langle V,V' \rangle),$$
which gives $|\alpha^{-1}(0)|= |\mV| (|\mV|-1)$ and $ |\alpha^{-1}(1)|= |\mV|$. It is also easy to see that $\mB$ is $\alpha$-regular.
Furthermore, we have
\begin{align*}
    \mW_0(\alpha)= N_{N/\ell-2,q^{\ell}}^{k-2}, \; \mW_1(\alpha)= N_{N/\ell-1,q^{\ell}}^{k-1},
\end{align*}
which, combined with Proposition~\ref{lem:upperbound}, directly implies the two bounds given in the theorem.
\end{proof}

The above directly imply the following bounds on the density of $\Bar{S}$-linear codes with minimum distance bounded from below.

\begin{corollary} 
Let $1 \le k \le N/\ell$ and $1 \le d < \infty$ be integers. We have
   \begin{align} \label{eq:upperBoundsublim}
    \delta_q^D(N,k,\ell, d) &\geq  1-  \frac{\textbf{v}_q^{D}(S_{q(\gamma)}^{n},d)N_{N/\ell-1,q^{\ell}}^{k-1}}{(q^{\ell(s-1)}(q^{\ell}-1))N_{N/\ell,q^{\ell}}^{k}},\\ \label{eq:upperBoundsublimup}
    \delta_q^D(N,k,\ell, d) &\le 1-  \frac{\textbf{v}_q^{D}(S_{q(\gamma)}^{n},d) N_{N/\ell-1,q^{\ell}}^{k-1} }{\bar\Theta(q^{\ell(s-1)}(q^{\ell}-1))N_{N/\ell,q^{\ell}}^{k} },
\end{align}
where $$\bar\Theta =  1+ (N_{N/\ell-1,q^{\ell}}^{k-1})^{-1}\left(\displaystyle \tfrac{\textbf{v}_q^{D}(S_{q(\gamma)}^{n},d)}{q^{\ell(s-1)}(q^{\ell}-1)} -1 \right) N_{N/\ell-2,q^{\ell}}^{k-2} .$$
\end{corollary}
\section{The asymptotics of $\mR$-linear codes}
In this section we present some general asymptotic results on the density function of codes in $\mR^{N}$ endowed with a translation-invariant metric $D$.

The asymptotic estimates of the $q$-binomial coefficient from \cite{b5} extend to the following equivalences.
\begin{lemma} \label{lem:Nestimates}
  Let $n \geq k \geq 0$ be integers. Then
  \begin{enumerate}
        \item[(i)]\label{lem:asymptoticsq}
  $  N_{n,q}^{k} \sim q^{sk(n-k)}$ \quad \textnormal{as $q \to \infty$},
      \item [(ii)]
      $N_{n,q}^{k} \sim q^{sk(n -k)}\prod_{i=1}^k \frac{q^i}{q^i-1}$ \quad \textnormal{as $n \to \infty$}.
  \end{enumerate}
\end{lemma}
\begin{theorem} \label{thm:asymq}
    Let $1 \leq s,  1 \leq k \leq N/\ell$ and $ 2 \leq d \leq n$ be integers. Further consider the sequence $(S_{q}^{n})_{q \in Q}$, where $S_{q}^{n}:= S_{q(\gamma)}^{n}$. Then
      \begin{align*} 
      \max\left\{ \liminf_{q \to \infty}\left( 1- \tfrac{\textbf{v}_q^{D}(S_{q}^n ,d)}{q^{s(N-k\ell+\ell)}}\right), 0\right\} \leq \liminf_{q \to\infty} \delta_q^D(N,k,\ell,d).
\end{align*}
 Moreover, if $\textbf{v}_q^{D}(S_{q}^{n},d) \in \Omega (q^{s(N-k\ell+\ell)})$ as $q \to \infty$, then 
    \begin{align} \label{eq:asymupperBound_q1}
      \limsup_{q \to \infty} \delta_q^D(N,k,\ell,d) \leq \limsup_{q \to \infty} \left( \tfrac{1}{1 + \tfrac{\textbf{v}_q^{D}(S_{q}^n ,d)}{q^{s(N-k\ell+\ell)}}}\right).
\end{align}
In particular,
\begin{align*}
  \lim_{q \to \infty}\delta_q^D(N,k,\ell,d)  = \begin{cases}
 1  &\textnormal{if $\textbf{v}_q^{D}(S_{q}^n ,d) \in o(q^{s(N-k\ell+\ell)})$,}  \\
    0  &\textnormal{if $\textbf{v}_q^{D}(S_{q}^n ,d) \in \omega(q^{s(N-k\ell+\ell)})$.} 
    \end{cases}
\end{align*}
    \begin{proof}
        From Lemma \ref{lem:asymptoticsq} we obtain 
        \begin{equation*}
            \frac{\textbf{v}_q^{D}(S_{q}^{n},d)N_{N/\ell-1,q^{\ell}}^{k-1}}{q^{\ell(s-1)}(q^{\ell}-1)N_{N/\ell,q^{\ell}}^{k}} \sim \frac{\textbf{v}_q^{D}(S_{q}^{n},d)}{q^{s(N-k\ell+\ell)}} \quad \textnormal{as}  \quad q \to  \infty.
        \end{equation*}
     Together with Equation \eqref{eq:upperBoundsublim} and the fact that $\delta_q^D(n,k,d) \ge 0$ the lower bound can be derived. Moreover, we have
     \begin{equation*}
         \bar\Theta \sim 1+ \frac{(\textbf{v}_q^{D}(S_{q}^{n},d)-q^{\ell s})}{q^{s(N-k\ell +\ell)}} \quad \textnormal{as}  \quad q \to  \infty ,
     \end{equation*}
     which, considering the upper bound \eqref{eq:upperBoundsublimup}, gives
     \begin{equation*}
         \tfrac{\textbf{v}_q^{D}(S_{q}^{n},d) N_{N/\ell-1,q^{\ell}}^{k-1} }{\bar\Theta(q^{\ell(s-1)}(q^{\ell}-1))N_{N/\ell,q^{\ell}}^{k} } \sim \tfrac{\textbf{v}_q^{D}(S_{q}^{n},d)}{q^{s(N-k\ell+\ell)}+\textbf{v}_q^{D}(S_{q}^{n},d)-q^{\ell s}}
     \end{equation*}
     as $q \to  \infty$. Since $\textbf{v}_q^{D}(S_{q}^{n},d) \in \Omega (q^{s(N-k\ell+\ell)})$ as $q \to  \infty$, taking the limit in the bound \eqref{eq:upperBoundsublimup} yields the upper bound  \eqref{eq:asymupperBound_q1}.
    \end{proof}
    \end{theorem} 
    Following the same arguments as in the proof of  Theorem \ref{thm:asymq} we obtain analogous 
    results for $n \to \infty$: 
    
    \begin{theorem} \label{thm:asymn}
    Let $q \in Q$ and $1 \leq s, 2 \leq d \leq n$ be integers. Further consider the sequence $(S_{q}^{n})_{n \geq 1}$, where $S_{q}^{n}:= S_{q(\gamma)}^{n}$ and let $(k(n))_{n \geq 1}$ be a sequence of integers where $1 \leq k(n) <n$ for all $n\geq 1$. Then 
      \begin{align*} 
      \max\left\{ \liminf_{n \to \infty}\left( 1- \tfrac{\textbf{v}_q^{D}(S_{q}^n ,d)}{q^{s(N-k(n)\ell+\ell)}}\right), 0\right\} \leq \liminf_{n \to\infty} \delta_q^D(N,k(n),\ell,d).
\end{align*}
 Moreover, if $\textbf{v}_q^{D}(S_{q}^{n},d) \in \Omega (q^{s(N-k(n)\ell+\ell)})$ as $n \to \infty$ then 
    \begin{align*} 
      \limsup_{n \to \infty} \delta_q^D(N,k(n),\ell,d) \leq \limsup_{n \to \infty} \left( \tfrac{1}{1 + \tfrac{\textbf{v}_q^{D}(S_{q}^n ,d)}{q^{s(N-k(n)\ell+\ell)}}}\right).
\end{align*}
In particular,
\begin{align*}
  \lim_{n \to \infty}\delta_q^D(N,k(n),\ell,d)  = \begin{cases}
 1  &\textnormal{if $\textbf{v}_q^{D}(S_{q}^n ,d) \in o(q^{s(N-k(n)\ell+\ell)})$,}  \\
    0  &\textnormal{if $\textbf{v}_q^{D}(S_{q}^n ,d) \in \omega(q^{s(N-k(n)\ell+\ell)})$.} 
    \end{cases}
\end{align*}
\end{theorem}
      
    In particular, we can fix a rate $R \in [0,1]$ such that $k(n)=RN/\ell$. Using the same setup, we obtain:
    \begin{corollary} \label{cor:rateasym}
   Let $N=nm, R \in [0,1], S_{q}^{n}:= S_{q(\gamma)}^{n}$. Then
\begin{align*}
 & \lim_{n \to \infty}\delta_q^D(n,RN/\ell,\ell, d)    = 
 & \begin{cases}
 1  &\textnormal{if $\textbf{v}_q^{D}(S_{q}^{n},d) \in o(q^{s((1-R)N+\ell)})$,} \\
    0  &\textnormal{if $\textbf{v}_q^{D}(S_{q}^{n},d) \in \omega(q^{s((1-R)N+\ell)})$.}
    \end{cases}
\end{align*}
    \end{corollary}
    The \emph{sphere covering} or \emph{Gilbert-Varshamov (GV) bound} states that there exist codes with minimum distance $d$ of cardinality at least $\frac{|S^{n}|}{\textbf{v}_q^{D}(S^n,d)}$. 
    Note that $|S^{n}| \sim |S_{q(\gamma)}^{n}|$ both, as $n \to \infty$ and as $q \to \infty$.  Hence, if we have that 
    $$\textbf{v}_q^{D}((S/\langle \gamma^{s-1}\rangle)^n,d) \in o(\textbf{v}_q^{D}(S^n,d)),$$
 as $n \to \infty$ or as $q \to \infty$,
    we can asymptotically approximate this bound by
    \begin{align} \label{eq:GV}
       |\mC| \geq \frac{|S^{n}|}{\textbf{v}_q^{D}(S^n,d)} \sim \frac{|S_{q(\gamma)}^{n}|}{\textbf{v}_q^{D}(S_{q(\gamma)}^n,d)} ,
   \end{align}
as either $q$ or $n$ tends to infinity, for some $\mC \subseteq S^{n}$ of minimum distance $d$.
     For linear codes of rate $R$ we can reformulate \eqref{eq:GV} to
$$ R \geq 1-\frac{1}{N}\log_{q^{s}}(\textbf{v}_q^{D}(S^n,d)) $$
$$ \iff \log_{q^{s}} (\textbf{v}_q^{D}(S^n,d)) \geq N(1-R).$$
In the following we study the asymptotic behaviour of codes achieving the bound \eqref{eq:GV} with respect to $q$ or $n$.

\begin{theorem}\label{thm:GVq}
    Let $2 \leq d \leq n$ be integers and $(S^{n},D)$ be a metric space, where $D$ is a translation-invariant metric on $S^{n}$, non-increasing under scalar multiplication. Further, consider the sequence $(S_{q}^{n})_{q \in Q}$, where $S_{q}^{n}:= S_{q(\gamma)}^{n}$. Let $\mC \subseteq S^{n}$ be a $\Bar{S}$-module of rate $R=1-\frac{1}{N}\log_{q^{s}}(\textbf{v}_q^{D}(S^n,d)) $,  
     chosen uniformly at random. If $\textbf{v}_q^{D}((S/\langle \gamma^{s-1}\rangle)^n,d) \in o(\textbf{v}_q^{D}(S^n,d)),$ as $q \to \infty$, then the probability that $\mC$ has minimum distance at least $d$ approaches $1$, for $q\rightarrow \infty$.
        \begin{proof}
       The statement follows from Theorem \ref{thm:asymq} and
        \begin{align*}
            \lim_{q \to \infty} \tfrac{\textbf{v}_q^{D}(S_{q}^{n},d)}{q^{s((1-R)N+\ell)}} & = \lim_{q \to \infty}q^{ s(\log_{q^{s}}( \textbf{v}_q^{D}(S_{q}^{n},d))-\frac{s((1-R)N+\ell)}{s})}\\
            &=\lim_{q \to \infty} q^{s((1-R)N-(1-R)N -\ell)} =0.
        \end{align*}
    \end{proof}
\end{theorem}

If we consider $n$ going to infinity, then the probability that a $\Bar{S}$-module, chosen uniformly at random attains the GV bound is upper bounded by $q^{s\ell}/(q^{s\ell}+1)$ and lower bounded by $(q^{s\ell}-1)/q^{s\ell}.$
However, if we add an $\varepsilon$-environment, we obtain:
\begin{theorem}
    In the setup of Theorem \ref{thm:GVq}
     let $\mC \subseteq S^{n}$ be a $\Bar{S}$-module of rate $R=1-\frac{1}{N}\log_{q^{s}}(\textbf{v}_q^{D}(S^n,d))-\varepsilon $, 
    chosen uniformly at random. If $\textbf{v}_q^{D}((S/\langle \gamma^{s-1}\rangle)^n,d) \in o(\textbf{v}_q^{D}(S^n,d)),$ as $n \to \infty$, then the probability that $\mC$ has minimum distance at least $d$ approaches $1$, for $n\rightarrow \infty$.
    \begin{proof}
    The statement follows from Corollary \ref{cor:rateasym} and
           \begin{align*}
            &\lim_{n \to \infty} \tfrac{\textbf{v}_q^{D}(S_{q}^{n},d)}{q^{s((1-R)N+\ell)}}= \lim_{n \to \infty} q^{s((1-R)N-(1-R)N -\varepsilon N -\ell)} =0.
        \end{align*}
    \end{proof}
\end{theorem}
Note that the rank metric fulfills the requirement $\textbf{v}_q^{D}((S/\langle \gamma^{s-1}\rangle)^n,d) \in o(\textbf{v}_q^{D}(S^n,d))$ for $q,n\rightarrow \infty$, as does the Hamming metric for $q\rightarrow \infty$. (For the definition of these metrics see the following two sections.) Our results hence extend the one on the asymptotic behaviour w.r.t. the GV bound from \cite{b1}, where the analog was shown for metrics stemming from additive weights (which do not include the rank metric).

\section{$\mR$-linear Hamming metric codes}
The \emph{Hamming weight} of $x \in S^{n}$ is given by $\omega^{H}(x):= |\{i \in [n] : x_{i} \neq 0 \}|$ and then the \emph{Hamming distance} for $x,y \in S^{n}$ is defined as $D^{H}(x,y):= \omega^{H}(x-y)$. Throughout this section, we are working in the metric space $(S^{n}, D^{H})$. Let $\mC$ be an $\Bar{S}$-module of length $n$ and $\mu_{\Bar{S}}(\mC)=k$ 
then the following Singleton-like bound is known for finite chain rings, see \cite{b2}:
$k \leq \frac{m}{\ell} (n-D^{H}(\mC)+1).$
Codes achieving this bound are called \emph{Maximum Distance with respect to Rank (MDR) codes}, to differentiate them from the usual MDS codes \cite{b2}.

In order to study the asymptotic density of Hamming-metric codes, we need the volume of the Hamming-metric ball in $S_{q(\gamma)}^{n}$ of radius $0 \leq r < \infty$ and its asymptotic estimates. 
Note that $|S|=q^{ms}$ and  $S$ has $q^{ms}-q^{m(s-1)}$ unit elements. Next we count the number of vectors up to Hamming weight $r-1$ having at least one unit, as 
\begin{equation*}
    \textbf{v}_{q}^{H}(S_{q(\gamma)}^{n},r)= \sum_{i=0}^{r-1}\binom{n}{i}iq^{m(s-1)}(q^{m}-1)(q^{ms}-1)^{i-1}.
\end{equation*}
The estimates
\begin{equation*} 
    \textbf{v}_{q}^{H}(S_{q}^{n},r) \sim 
  \binom{n}{r-1}(r-1)q^{ms(r-1)}  
  \end{equation*}
 as $q \to \infty$ and
  \begin{equation*}
        \textbf{v}_{q}^{H}(S_{q}^{n},r) \sim  \binom{n}{r-1}(r-1) (q^{ms}-1)(q^{ms}-q^{m(s-1)})^{r-2}
  \end{equation*}
  as $n \to \infty$, easily follow.
  We can now derive asymptotic results for the density of MDR codes.
\begin{theorem}
Let $ d \geq 2$ be an integer. 
\begin{itemize}
\item[(i)]
Let $n \ge 2$ be an integer. Then we have
$$\lim_{q \to \infty}\delta_{q}^{H}(N,m/\ell(n-d+1),\ell,d)=1.$$
\item[(ii)]
Let $q \in Q$, then
$$\lim_{n \to \infty}\delta_{q}^{H}(N,m/\ell(n-d+1),\ell,d)=0.$$
\end{itemize}
\begin{proof}
One easily gets $\textbf{v}_{q}^{H}(S_{q}^{n},d) \in o(q^{sm(d-1)+s\ell})$ as $q \to  \infty$, which together with Theorem  \ref{thm:asymq} gives the first statement of the theorem. The second statement follows from $\textbf{v}_{q}^{H}(S_{q}^{n},d) \in \omega(q^{sm(d-1)+s\ell})$ as $n \to \infty$ and Theorem \ref{thm:asymn}.
 \end{proof}
\end{theorem}

\section{$\mR$-linear rank metric codes}

We adopt the vector representation of rank metric codes as introduced in \cite{b3}.
A rank metric code is a subset of the $S$-module $S^{n}$ endowed with the \emph{rank distance}, i.e., 
$D^{\textnormal{rk}}(x,y) := \textnormal{rk}(x-y),$
where $\textnormal{rk}(x) := \mu_{\mR}(\langle x_1, \dots , x_n\rangle )$ denotes the \emph{rank} of $x= (x_1, \dots , x_n) \in S^{n}$. 

Let $\mC \subseteq S^{n}$ be a code of minimum rank distance $d$, then the following Singleton bound is given in \cite{b3}:
\begin{equation} \label{eq:Singletonrank}
    |\mC| \leq q^{s \cdot\max\{m,n\}(\min\{m,n\}-d+1)}.
\end{equation}
A rank metric code meeting the bound with equality is called an \emph{MRDR (Maximum Rank Distance with respect to Rank) code}.

Let $\mC \subseteq S^{n}$ be a free $\Bar{S}$-module such that $\mu_{\Bar{S}}(\mC)=k$ and $D^{\textnormal{rk}}(\mC)=d$, then the Singleton bound \eqref{eq:Singletonrank} can be restated as
\begin{align} \label{eq:Singletonrankmax}
    \ell k \leq \max\{m,n\}(\min\{m,n\}-d+1).
\end{align}

To analyze the asymptotic density of codes with largest possible rank over $\Bar{S}$, we first derive the volume of the ball in the rank metric and its asymptotic estimates, similarly to \cite[Proposition 15]{b1}.
\begin{proposition}
    Let $0 \leq r-1 \leq \min\{m,n\}$, then
    \begin{equation*}
         \textbf{v}_{q}^{\textnormal{rk}}( S_{q(\gamma)}^{n},r):= \sum_{i=0}^{r-1}N_{n,q}^{i}q^{(s-1)mi}\prod_{j=0}^{i-1}(q^{m}-q^{j}).
    \end{equation*}
    \begin{proof}
        For a fixed free module $U \subseteq \mR^{n}$ with $\mu_{\mR}(U) =k$, the number of matrices $A \in \mR^{n \times m}$ with $\textnormal{colsp}(A)=U$ equals the number of matrices $B \in \mR^{k \times m}$ of rank $k$, where by abuse of notation, a matrix that generates a free module of rank $k$ is called a matrix of rank $k$.
        The number of matrices of rank $k$ is given by
        $$q^{(s-1)mk}\prod_{j=0}^{k-1}(q^{m}-q^{j}),$$
       i.e., the number of ways to draw $k$  vectors from $\mR_{\gamma}^{m}$, such that they do not lie in the span of the previously picked vectors. 
The number of free modules in $\mR^{n}$ having rank $k$ is $N_{n,q}^{k}$.
    \end{proof}
\end{proposition}
Using the estimates from Lemma \ref{lem:Nestimates} we obtain for $0 \leq r-1 \leq \min\{m,n\}$ the asymptotic estimate
$$ \textbf{v}_{q}^{\textnormal{rk}}( S_{q}^{n},r) \sim q^{s(r-1)(m+n-(r-1))},$$
as $q \to \infty$ and for $0 \leq r-1 \leq m$ the estimate
\begin{align} \label{eq:estimaterankn}
     \textbf{v}_{q}^{\textnormal{rk}}( S_{q}^{n},r) \sim \dstirling{m}{r-1}_{q}q^{(r-1)((s-1)(m+n-(r-1))+n)},
\end{align}
as $n \to \infty$.

If $n > m$, codes attaining the largest possible rank over $\Bar{S}$, according to \eqref{eq:Singletonrankmax} are not necessarily MRDR.
For this reason we state the results for $n>m$ separately. The proof is an application of \ref{thm:asymq}, for details we refer to \cite[Theorem 5.7]{b5}.
\begin{theorem} 
Let $m,\ell \geq 1, 2 \leq d \leq n$ be integers and let $k=\lfloor \max\{m,n\}(\min\{m,n\}-d+1)/\ell \rfloor$. Define $\theta := (d-1)(\min\{m,n\}-d+1)$.
\begin{itemize}
    \item[(i)] If $m \ge n$ then we have
    \begin{align*}
\lim_{q \to \infty} \delta_q^\textnormal{rk}(N,k,\ell,d) =
\begin{cases}
     1  &\textnormal{if $\ell > \theta$,} \\
    0  &\textnormal{if $\ell < \theta$}
    \end{cases}
\end{align*}
and if $\ell =\theta$, then $\lim_{q \to \infty} \delta_q^\textnormal{rk}(N,k,\ell,d) \le 1/2.$
    \item[(ii)] If $n > m$ then we have
   \begin{align*}
\lim_{q \to \infty} \delta_q^\textnormal{rk}(N,k,\ell,d) =
\begin{cases}
     1 &\textnormal{if $\theta -r < \ell$,} \\
    0 &\textnormal{if $\theta -r > \ell$,}
    \end{cases}
\end{align*}
where $r:=\ell(\lceil n(d-1)/\ell \rceil -n(d-1)/\ell)$.
Moreover, if $\ell = \theta -r$ then $\lim_{q \to \infty} \delta_q^\textnormal{rk}(N,k,\ell,d) \le 1/2$.
\end{itemize}
\end{theorem}
If we let the  length $n$ go to infinity, we are again in the setting where codes of the largest possible dimension are not necessarily MRDR.
\begin{theorem}
    Let $q \in Q$, fix $2 \leq d \leq m$ and let $k(n):= \lfloor n(m-d+1)/\ell \rfloor$ for all $n \geq 2$. Then 
       \begin{align*}
\limsup_{n \to \infty} \delta_q^\textnormal{rk}(N,k,\ell,d) \leq \frac{1}{1+\dstirling{m}{d-1}_{q}\frac{q^{(d-1)((s-1)(m-(d-1))}}{q^{2s\ell}}}.
\end{align*}
\begin{proof}
  Applying Theorem \ref{thm:asymn} and the estimate \eqref{eq:estimaterankn} as $n \to \infty$ we obtain
  \begin{align*}
      \lim_{n \to  \infty} \frac{\textbf{v}_{q}^{\textnormal{rk}}( S_{q}^{n},d)q^{sk(n)\ell}}{q^{s(N+\ell)}} &= \lim_{n \to  \infty}\dstirling{m}{d-1}_{q}\frac{q^{(d-1)(s(m-(d-1))}}{q^{s\ell(1+(\lceil n(d-1)/\ell \rceil -n(d-1)/\ell))}}\\
      & \sim \dstirling{m}{d-1}_{q}\frac{q^{(d-1)((s-1)(m-(d-1))}}{q^{2s\ell}},
  \end{align*}
  where we used that $k(n)= \lfloor n(m-d+1)/\ell \rfloor = N/\ell -\lceil n(d-1)/\ell \rceil$ and $0 \leq  (\lceil n(d-1)/\ell \rceil - n(d-1)/\ell ) \leq 1$.
\end{proof}
\end{theorem}
\section{Concluding Remarks}
We showed that linear codes over a finite chain ring endowed with the rank or the Hamming metric exhibit a similar asymptotic behaviour as in the finite field case, if we let $q$ or $n$ go to infinity. Under the assumption that $\textbf{v}_q^{D}((S/\langle \gamma^{s-1}\rangle)^n,d) \in o(\textbf{v}_q^{D}(S^n,d)),$ codes attaining (an $\varepsilon$-environment of) the Gilbert-Varshamov bound are asymptotically dense as $q \to \infty$ (respectively as $n \to \infty$). Over finite fields this requirement is always fulfilled, as then the zero-vector is the only element containing only non-units. 
Furthermore, we showed in which cases optimal codes in the Hamming or the rank metric are dense or sparse.

We remark that, using the Chinese Remainder Theorem, one can extend this procedure to codes over general finite principal ideal rings, to determine their asymptotic behaviour.


\begin{thebibliography}{00}
\bibitem{b1} E. Byrne,  A.L. Horlemann, K. Khathuria and V. Weger, "Density of free modules over finite chain rings." Linear Algebra and its Applications 651 (2022): 1-25.

\bibitem{b2} S.T. Dougherty and K. Shiromoto, ``MDR codes over $\mathbb{Z}_{k}$,'' IEEE Transactions on Information Theory, 46(1):265-269,2000


\bibitem{b3} H.T. Kamche and C. Mouaha, ``Rank metric codes over finite principal ideal rings and applications,'' IEEE Transactions on Information Theory, 65.12 (2019): 7718-7735.

\bibitem{b5} A. Gruica, A.L. Horlemann, A. Ravagnani and N. Willenborg, ``Densities of Codes of Various Linearity Degrees in Translation-Invariant Metric Spaces,''  arXiv preprint arXiv:2208.10573, August 2022.

\bibitem{b6}
B.R. MacDonald, "Finite rings with identity." Marcel Dekker Incorporated, 1974, vol. 28.

\bibitem{b11} A. Neri, A.L. Horlemann-Trautmann, T. Randrianarisoa and J. Rosenthal, ``On the genericity of maximum rank distance and Gabidulin codes,'' Designs, Codes and Cryptography 86.2 (2018): 341-363.

\bibitem{heide}
J. Antrobus and H. Gluesing-Luerssen, "Maximal Ferrers diagram codes: constructions and genericity considerations." IEEE Transactions on Information Theory 65.10 (2019): 6204-6223.

\bibitem{alb}
A. Gruica,  and A. Ravagnani, "Common complements of linear subspaces and the sparseness of MRD codes." SIAM Journal on Applied Algebra and Geometry 6.2 (2022): 79-110.

\bibitem{honold2000linear}
T. Honold and  I. Landjev, "Linear codes over finite chain rings",
 The electronic journal of combinatorics,
7 (2000).
\bibitem{lee}
E. Byrne and V. Weger, "Bounds in the Lee Metric and Optimal Codes."  Finite Fields and Their Applications, 2022.

\bibitem{shannon48}
C.E. Shannon,  "A mathematical theory of communication." The Bell system technical journal 27.3 (1948): 379-423.

\bibitem{b13}
C. Feng, D. Silva and F.R. Kschischang: "An algebraic approach to physical-layer network coding." IEEE
Trans. Inf. Theory 59(11), 7576–7596 (2013).
\bibitem{b14}
M.P. Wilson, K. Narayanan, H.D. Pfister and A. Sprintson: "Joint physical layer coding and network coding
for bidirectional relaying." IEEE Trans. Inf. Theory 56(11), 5641–5654 (2010)
\end{thebibliography}
\end{document}